%% file: neurips_2022.tex
\documentclass{article}

\usepackage[final,nonatbib]{neurips_2022}

\usepackage[utf8]{inputenc} % allow utf-8 input
\usepackage[T1]{fontenc}    % use 8-bit T1 fonts
\usepackage{hyperref}       % hyperlinks
\usepackage{url}            % simple URL typesetting
\usepackage{booktabs}       % professional-quality tables
\usepackage{amsfonts}       % blackboard math symbols
\usepackage{nicefrac}       % compact symbols for 1/2, etc.
\usepackage{microtype}      % microtypography
\usepackage{xcolor}         % colors
\usepackage{wrapfig}
\usepackage{booktabs} % For formal tables
\usepackage[ruled]{algorithm2e} % For algorithms

\usepackage{amsmath}
\usepackage{graphicx}
\usepackage{cleveref}
\usepackage{amsthm}
\usepackage{verbatim}
\usepackage{subfigure}
\usepackage{tikz}
\usepackage{pgfplots}
\usepackage{setspace}
\usepackage{enumitem}
\pgfplotsset{compat=newest}

\newtheorem{theorem}{Theorem}[section]
\newtheorem{lemma}[theorem]{Lemma}

\newtheorem{corollary}[theorem]{Corollary}

\newtheorem{definition}[theorem]{Definition}

\input{macros}
\newcommand{\eat}[1]{}

\title{Online Algorithms for the Santa Claus Problem}

\author{%
  MohammadTaghi Hajiaghayi \\
  Department of Computer Science\\
  University of Maryland\\
  College Park, MD \\
  \texttt{hajiagha@umd.edu} \\
  \And
  MohammadReza Khani \\
  Microsoft Bing Ads \\
  Microsoft Corporation \\
  Redmond, WA \\
  \texttt{khani87@gmail.com}
  \And 
  Debmalya Panigrahi \\
  Department of Computer Science \\
  Duke University \\
  Durham, NC \\
  \texttt{debmalya@cs.duke.edu}
  \And
  Max Springer \\
  Department of Mathematics\\
  University of Maryland\\
  College Park, MD \\
  \texttt{mss423@umd.edu} \\
}

\begin{document}
% \raggedbottom
% \flushbottom
\maketitle

\input{abstract}

% -------------------------------------
\section{Introduction} \label{sec.intro} \input{ms-introduction} % \input{introduction}

% \subsection{Preliminaries} \label{sec.prelim}% Notational definitions only, bring random order algo stuff up here
% \input{prelim-scalar}

% -------------------------------------
\section{Online Algorithm for the Santa Claus Problem in the Random Order Model}\label{sec.algo}
\input{random_condensed}

\subsection{Online Rounding Algorithm}\label{sec:rounding}\input{rounding}

% -------------------------------------
\section{Random Order Lower Bound}\label{sec.bounds}\input{lower}

%\section{Conclusion}\label{sec.con}\input{conclusion}
\input{conclusion}

\nocite{Kawase_2021}
\nocite{Alaei_2013,Azar_1998,devanur2009adwords,Devanur_2011,Devanur_2012,Epstein_2011,Epstein_2018,Goel_2008,Liu_2021,Mehta_2007,Panigrahi_2012}

% \clearpage

\eat{
\section{Broader Impact}
Strong algorithmic results to fair allocation problems can have tremendous impacts on key societal issues such as allocation of medical services or approaches to voter redistricting. In a highly relevant case, fair allocation decisions have played a major role in the distribution of COVID-19 vaccines and mathematically founded approaches would produce a highly equitable result for the afflicted regions. Recent work has also explored the implications of mathematical approaches to state voting redistricting to combat gerrymandering and political biases, to which fair division provides a logical framework to improve.}

\section{Acknowledgements and Disclosure of Funding}
MohammadTaghi Hajiaghayi was upported in part by NSF CCF grants 2114269 and 2218678. Debmalya Panigrahi was supported in part by NSF grants CCF-1750140 (CAREER) and CCF-1955703 and ARO grant W911NF2110230. Max Springer was supported by the National Science Foundation Graduate Research Fellowship Program under Grant No. DGE 1840340. Any opinions, findings, and conclusions or recommendations expressed in this material are those of the author(s) and do not necessarily reflect the views of the National Science Foundation.

\bibliographystyle{acm}
\bibliography{neurips_2022}

\clearpage
\section*{Appendix}
\appendix
\section{Omitted Proofs} \label{sec:omitted}\input{omitted_proofs}

\end{document}

%% file: macros.tex
\newcommand{\todo}[1]{{\textcolor{red}{{\bf TODO:} #1}}}

\newcommand{\bY}{\mathbf{Y}}
\newcommand{\bZ}{\mathbf{Z}}
\newcommand{\bv}{\mathbf{v}}
\newcommand{\bV}{\mathbf{V}}
\newcommand{\bVo}{\overline{\mathbf{V}}}
\newcommand{\ve}{\varepsilon}
\newcommand{\bx}{\mathbf{x}}
\newcommand{\xopt}{\overline{\mathbf{x}}}

\newcommand{\alg}{\textsc{\sc ALG}\xspace}
\newcommand{\opt}{\text{\sc OPT}\xspace}
\newcommand{\e}{\varepsilon}
\renewcommand{\d}{\delta}
\renewcommand{\Pr}{\text{\bf Pr}\xspace}
\newcommand{\Ex}{\mathbb{E}\xspace}

%% file: abstract.tex
% Abstract. Note that this must come before \maketitle.
\begin{abstract}
The Santa Claus problem is a fundamental problem in {\em fair division}: the goal is to partition a set of {\em heterogeneous} items among {\em heterogeneous} agents so as to maximize the minimum value of items received by any agent. In this paper, we study the online version of this problem where the items are not known in advance and have to be assigned to agents as they arrive over time. If the arrival order of items is arbitrary, then no good assignment rule exists in the worst case. However, we show that, if the arrival order is random, then for $n$ agents and any $\e > 0$, we can obtain a competitive ratio of $1-\e$ when the optimal assignment gives value at least $\Omega(\log n / \e^2)$ to every agent (assuming each item has at most unit value). We also show that this result is almost tight: namely, if the optimal solution has value at most $C \ln n / \e$ for some constant $C$, then there is no $(1-\e)$-competitive algorithm even for random arrival order.

\eat{
We consider the classic ``Santa Claus'' allocation problem for indivisible items. More formally, $m$ items arrive one-by-one with inherent valuations and must be irrevocably assigned to one of $n$ agents with an objective of fairness across agents. In this paper, we consider fairness to be \emph{maximizing} the value of the bundle assigned to the \emph{minimally} satisfied agent. This problem has been widely explored in the offline setting, however we provide a pioneering $(1-\ve)$-competitive ratio for both the \textsc{iid} and random-order arrival models using an approximately greedy algorithm for any $\ve > 0$. Most notably, we provide a polynomial-time algorithm which yields a near optimal allocation subject to a novel lower bound of $\Omega \left(\frac{\log n}{\ve} \right)$ on the offline optimal solution in minimal computation time. In addition to the tight competitive ratio and matching lower bound, we provide multiple lower bounding results for the adversarial context. These results provide concrete conditions under which a nearly optimal result can be achieved for the open ended problem.
\vspace{5mm} \\
\todo{Add adversarial competitive ratio result}
}
\end{abstract}

%% file: ms-introduction.tex
Fair allocation of resources is one of the central themes of algorithmic fairness and game theory. In fact, the theory of fair division has its roots in mathematics going back to as early as 1948~\cite{steihaus1948problem}. In the general setting, this problem comprises a set of items that must be divided among a set of agents in an egalitarian manner, where each agent has a (possibly non-uniform) valuation for each item. A natural objective to capture the goal of fair division is to maximize the minimum total value of items received by any agent. This gives rise to the famous ``Santa Claus problem'' that we describe below.

In the Santa Claus problem, originally described by Bansal and Sviridenko in 2006 \cite{Bansal_2006} (although it was studied under different names or assumptions prior to this), the Santa Claus is said to have a set of $m$ presents to be distributed equitably among $n$ children. Each child $i\in [n]$ has some arbitrary non-negative value $v_{ij}$ for present $j\in [m]$. Santa's goal is to distribute the presents in a way that makes the least satisfied child maximally satisfied. More formally, this means that the assignment seeks to maximize the minimum total value of the presents received by any child, where the total value of presents received by a child is the sum of her values for the presents that she received. The Santa Claus problem can be formalized as the following integer program:
  \[ \max\bigg\{ \min_{i\in [n]} \sum_{j=1}^m v_{ij} x_{ij}
  \mid \sum_{i=1}^n x_{ij} \le 1 \;\forall j \in [m], \; x \in \{0, 1\}^{mn}
  \bigg\}. \]
\eat{

\begin{gather*}
\max \lambda \\
\begin{aligned}
\textup{s.t.}\quad &\forall i \in [n] &\quad \sum_{j=1}^m v_{ij} \cdot x_{ij}  \geq  \lambda \\
            &\forall j \in [m] &\quad \sum_{i=1}^n x_{ij} \le 1\\
            &\forall i\in [n], j\in [m] &\quad x_{ij} \in \{0, 1\}
\end{aligned}
\end{gather*}
}

There is substantial literature going back more than 50 years that studies variants of this problem in the offline setting (see related work).
In many practical situations, however, the set of items to be allocated is not known in advance. For example, in online advertising, ad-space providers will receive monetary bids from competing agents for the display of their advertisements in real-time for an available space on a webpage \cite{Buchbinder_2007,Blum_2006,Hajiaghayi_2007,Zhou_2008}. The provider must then make irrevocable decisions as to which advertiser's bid to accept based only on knowledge of prior allocations and the current bids for the available space \cite{balseiro2019learning,conitzer2021multiplicative}. Beyond advertising, online allocation procedures have been useful in the study of donation distribution, wireless charging networks, organ donor matching, etc (see \cite{aleksandrov2020online} for a survey of these applications). 

\begin{wrapfigure}{R}{0.5\textwidth}
    \centering
	\includegraphics[width=0.38\textwidth]{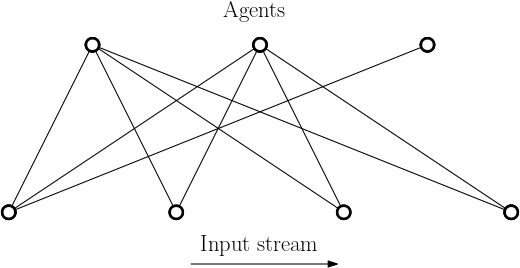}
	\caption{Example online problem instance. Fixed agents are presented in the top row and the arriving items are arranged from left to right in correspondence with their arrival order. An edge indicates a nonzero item valuation for a given agent.} %\todo{enlarge font}
	\label{fig:hard_ex}
\end{wrapfigure}

Motivated by these applications, we consider the {\em online} Santa Claus problem in this paper. In this setting, the items arrive in an online sequence and must be allocated to one of the agents immediately upon arrival. As in the offline problem, our goal is to maximize the minimum total value among all agents. %While the problem may seem like a trivial extension, in practice it increases the complexity tremendously. 
To illustrate the problem, consider the simple example in Figure \ref{fig:hard_ex} on the right where edges represent unit value of an agent for an item. At the time of the first arriving (leftmost) item, all agents can be matched to this item and therefore with probability $\frac{1}{3}$ any agent will receive it. However, as we continue forward with the input stream, we see in retrospect that the only nonzero max-min solution corresponds to the case where the first item was allocated to the rightmost agent (as it is the only item for which they have a nonzero value). %In the offline instance, it is very easy to see this solution, but the incremental reveal of information requires both sophisticated algorithms and assumptions on the input instance that diminish the possibility of such complexity.

Following standard terminology for online algorithms (see, e.g., \cite{borodin2005online}), we define the {\em competitive ratio} of an online algorithm as the minimum ratio between the value of the (maximization) objective in the algorithm's solution to that in the optimal (offline) solution in hindsight. We furthermore discuss the \emph{additive regret} as the additive loss factor of our algorithm. More formally, we say our algorithm, \alg, has competitive ratio $c$ and additive regret $b$ if $\alg \geq c \cdot \opt{} - b$.

Prior work on the max-min objective in the online setting required various relaxations of the problem, such as allowing for some reordering in the allocation process \cite{Epstein_2011}, restricting the number of agents \cite{he2005optimal,tan2006semi,wu2014optimal}, or allowing migration of items after assignment \cite{chen2011line}. This is because of two reasons. First, even in the offline setting, there remains a significant gap between the best upper and lower bounds on the approximation ratio of the Santa Claus problem, and bridging this gap is a major open problem. Second, as we will soon see, there is a simple construction for the online problem that shows the competitive ratio cannot be better than $n$. To bypass these bottlenecks, our first assumption in this paper is that the items arrive in {\em random order}. This is a standard assumption that has been used to simplify many related online problems \cite{babaioff2007knapsack,devanur2009adwords,feldman2009online,feldman2009online2,goel2008online,karande2011online,kleinberg2005multiple}. But, even with this assumption, we show that obtaining a competitive solution is impossible in general {\em for small problem instances}. This motivates our second assumption: that the objective value of the optimal solution is sufficiently large (with respect to the values of individual items). %In fact, we obtain the precise threshold on the optimal objective in the random order setting. 
With these two assumptions, we give an algorithm that obtains a competitive ratio of $(1-\e)$ for any $\e > 0$. We note that using standard techniques, the assumption about the optimal objective being sufficiently large can be replaced by a corresponding additive regret in the competitive ratio.

We now formally define the two online input models that we consider in this paper: \emph{adversarial} and \emph{random order} input.
\vspace{-1mm}
\begin{definition}[Adversarial Input]
An adversary selects the value vector $v \in [0,1]^n$ of each arriving item for all the agents, as well as the order in which these vectors arrive.
\end{definition}
\vspace{-1mm}
\begin{definition}[Random-Order Input]
An adversary selects the value vector $v \in [0,1]^n$ of each arriving item for all the agents, but these vectors are randomly permuted to determine their arrival order.
\end{definition}
Note that in the literature, the independent and identically distributed (i.i.d.) input model is also often studied for related problems \cite{goel2008online,karp1990optimal,mahdian2011online,karande2011online,Agrawal_2014,kesselheim2014primal,seiden2002online,molinaro2014geometry}. In this model, the adversary picks a distribution over inputs that is unknown to the algorithm and arriving items are sampled i.i.d. from this distribution. The random order model is {\em stronger} than the i.i.d. model in the sense that any algorithmic result for the random order model automatically extends to the i.i.d. setting as well. This includes the algorithmic results that we obtain in this paper for the random order arrival model. %As a result, any algorithm that works for the random order input will also work on an i.i.d. input. We thus assume these results, but do not formally consider them in the main text. 

\subsection{Problem Definition}
We here introduce the notation that we will use in the rest of the paper. Let $\bv^1, ..., \bv^m \in [0,1]^n$ denote the input sequence of items arriving in random order where $v^t_i$ is the value of the $t$-th item to the $i$-th agent. We additionally denote by $\bx^1, ... , \bx^m \in [0,1]^n$ the fractional allocation of each item by the algorithm.
We further let the corresponding allocations of a fixed (offline) optimal solution be denoted as $\xopt^1, ..., \xopt^m \in [0,1]^n$ and let $\opt{}$ denote the max-min objective value of the optimal solution. 
For simplicity, we slightly abuse notation by letting $\bV^t = (v^t_1 x^t_1, \dots, v^t_n x^t_n)$ and $\bVo^t = (v^t_1 \overline{x}^t_1, \dots, v^t_n \overline{x}^t_n)$ for $t \in [m]$. Our algorithm thus seeks to maximize the coordinate-wise minimum of $\sum_{t=1}^m \bV^t$.

\subsection{Our Contributions} \label{sec.our_work}
First, we give a simple construction to show that if the arrival order of the items is adversarial, then the best competitive ratio that can be achieved is only $1/n$. (In fact, we also match this competitive ratio using a simple algorithm in the supplementary material.)
\begin{theorem}[Adversarial Input] \label{thm:ad-upper}
    In the adversarial setting, no algorithm can obtain a competitive ratio better than $1/n$.
\end{theorem}

This motivates us to consider the \emph{random order model}, where an adversary again selects the set of items but they are then presented in random permutation order. For this setting, we give an algorithm that obtains a {\em fractional} assignment that is nearly optimal: %This result is achieved by a smoothing of the minimum function, which we will allocate with respect to, as well as a ``restart'' on the allocation values halfway through the input stream to alleviate the correlations that arise between items. Crucially, we also demonstrate a necessary dependence on the input instance size to allow for a $(1-\ve)$-competitive ratio. This result is stated precisely as follows:

\begin{theorem}[Random Order: Algorithm] \label{thm:main_ROM}
    For any $\ve > 0$, there is an online fractional algorithm for the Santa Claus problem that has a competitive ratio of $(1 - \ve)$ in the random order input model under the assumption that $\opt{} \geq \Omega \left( \log n / \ve^2 \right)$.
\end{theorem}

We further show that, through randomized rounding, we can give an \emph{integral} allocation that retains the near optimality of this fractional allocation.

% \alert{Can't we say something about integral allocations by using randomized rounding since $\opt$ is large?}

%Our complimentary result comes in the form of an impossibility result for algorithms in the random order input model. 
Finally, we show that the lower bound on the value of $\opt$ in the above theorem is {\em necessary}:

\begin{theorem}[Random Order: Impossibility Result] \label{thm:main_lower}
    For any $\ve \in (0, 1)$, there is no online algorithm for the Santa Claus problem in the random order input model that has a competitive ratio of $(1-\ve)$ when $\opt < C\cdot \frac{\ln n}{\ve}$ for some (absolute) constant $C > 0$.
\end{theorem}

The reader will note that the lower bound on $\opt$ is precise as a function of $n$, but there is a slight mismatch between the upper and lower bounds as a function of $\ve$ -- bridging this gap is an interesting open question.

We summarize our results for the {\em online} Santa Claus problem in Table~\ref{tab:one}.

% \usepackage{graphicx}
\begin{table}[t] \label{tab:one}
\centering
\caption{Online Results for the Santa Claus Problem}
{
% \resizebox{\textwidth}{!}{%
\begin{tabular}{ccc}
\multicolumn{3}{c}{\textbf{Our Results}}                               \\ \hline
\textit{Input Model} & \textit{Algorithm} & \textit{Competitive Ratio} \\ \hline
Adversarial  & \textsc{Random}  & $\left( \frac{1-\varepsilon}{n} \right)\textsc{Opt} - O \left( \frac{n \log n}{\varepsilon^3} \right)$ \\ \hline
Random Order & \textsc{GreedyWR} & $(1-\varepsilon)\textsc{Opt} - O\left( \frac{\log n}{\varepsilon} \right)$                 
\end{tabular}%
}
\end{table}

\subsection{Related Work}\label{sec.related}
The general case of the Santa Claus problem was initially explored (under a different name) in the field of algorithmic game theory for the fair allocation of goods \cite{lipton2004approximately}. By studying the assignment LP of \cite{lenstra1990approximation} for the dual ``makespan'' problem, Bezakova and Dani \cite{bezakova2005allocating} derived an additive approximation of $\max_{ij} v_{ij}$, i.e., the objective in the algorithm's solution is at least $\opt - \max_{ij} v_{ij}$ where \opt{} is the objective value of the optimal solution. They also extended the hardness result on the dual makespan minimization problem~\cite{lenstra1990approximation} to demonstrate that the Santa Claus problem is NP-hard and cannot be approximated to a factor better than 2. Later, Bansal and Sviridenko~\cite{Bansal_2006} demonstrated that the integrality gap of the configuration LP for this problem is $\Omega(\sqrt{n})$, while Asadpour and Saberi~\cite{Asadpour_2010} complimented this result with a $O(\sqrt{m} \log^3 m)$ upper bound for the same LP relaxation. To date, the best algorithmic result for the Santa Claus is an $\tilde{O}(n^\ve)$-approximate algorithm, where $\ve = \Omega(\log \log n / \log n)$, in quasi-polynomial time obtained by Chakrabarty, Chuzhoy, and Khanna~\cite{ChakrabartyCK09}.

For the special case of \emph{restricted assignment}, i.e. $v_{ij} \in \{0,v_j\}$, Bansal and Sviridenko \cite{Bansal_2006} provided an $\Omega \left(\log \log \log m / \log \log m \right)$-approximate algorithm that relies on rounding a configuration LP. Later, Feige gave a non-constructive proof that this LP relaxation was within a constant factor of $\opt{}$~\cite{feige2008allocations}. Asadpour, Feige, and Saberi~\cite{Asadpour_2012} made this constructive, obtaining a $1/4$-approximation via a rounding algorithm based on local search, but the algorithm is not known to converge in polynomial time. Further work has since improved this constant factor \cite{jansen2018note,cheng2018restricted,davies2019tale,annamalai2016combinatorial}, improved upon the running time \cite{cheng2019restricted}, and extended the setting beyond additive valuations \cite{bamas2020submodular}. 

{\bf Online Assignment.} The study of online assignment is expansive, but much of the classical work is for adversarial arrival order. Even in the random order setting, a broad range of problems have been considered in recent years including the secretary problem \cite{babaioff2007knapsack,kleinberg2005multiple}, AdWords \cite{devanur2009adwords,feldman2009online,goel2008online}, online matching \cite{feldman2009online2,karande2011online}, online packing \cite{Gupta_2014,feldman2010online,kesselheim2014primal}, online scheduling \cite{Molinaro_2017,Liu_2021}, etc. One example of max-min online assignment in the random order setting is the work of Gollapudi and Panigrahi~\cite{Gollapudi_2014} who considered revenue maximization with fairness objectives. %The work presents a near optimal $(1-\ve)$-competitive algorithm for the i.i.d. input setting but, despite this positive result, the authors do not extend their results to the more complex and practical \emph{random order model}, a novel result presented in our paper. Furthermore, while the authors make the assumption that the optimal offline solution is lower bounded as $\Omega \left(\log n / \ve \right)$ they do not provide a proof on the necessity of such an assumption, a point that we make clear in the present work.
Another related work is that of Molinaro~\cite{Molinaro_2017} for the dual min-max objective, who builds on prior work leveraging the experts framework from online learning \cite{Gupta_2014} to give algorithms that simultaneously perform well in the adversarial and random-order settings. %, balancing the hardness of the two distinct input streams to achieve a ``best of both worlds'' result. In achieving this, the authors relax the maximum function to a general $L^p$ norm and analyze the problem as it depends on the parameter $p$. 
%We here extend their results by emulating the proof technique for the \emph{concave} objective function of a smoothed minimum. Further work by \cite{Liu_2021} demonstrated that the greedy allocation procedure is necessarily optimal for the makespan problem without giving the exact competitive ratio of this algorithm. As we will demonstrate in our competitive ratio analysis, our results are effectively the dual compliment to the prior work. While in general minimization and maximization problem solutions do not align, we demonstrate that in this problem instance the two align well. 
A third line of work relevant to our paper is that of online packing problems in the random arrival order (e.g., \cite{feldman2010online,Agrawal_2014,kesselheim2014primal,seiden2002online,molinaro2014geometry}. In particular, the results of Agrawal~{\em et al.}~\cite{Agrawal_2014} have a similar flavor to ours: they obtain $(1-\ve)$-competitiveness assuming a large enough optimal value for the online packing problem in the random order setting. %This result is derived by a careful worst-case example construction, a method that we closely follow in the construction of our impossibility results for the current problem. %, though we only demonstrate the looser $\Omega \left( \frac{\log n}{\ve} \right)$ bound.

%\noindent \textbf{Matching.} Online bipartite matching was originally presented by Karp et al. \cite{karp1990optimal} and involves a bipartite graph $G(U,V,E)$ where one side is known and the other side arrives in an online fashion. As the vertices arrive, its neighbors are revealed and the vertex must be matched to one of the available neighbors (if they exist). The objective is to maximize such a matching, and as a result the problem can effectively be seen as a sub-problem on the Santa Claus problem with unit value items and a capacity on each agent. The problem has been thoroughly studied in both the adversarial \cite{aggarwal2011online,birnbaum2008line,goel2008online,karp1990optimal} and random order inputs \cite{feldman2009online3,goel2008online,jaillet2014online,karande2011online,karp1990optimal,mahdian2011online}.% and has further been extended to the case of \emph{weighted} vertex matching where graph edges have a weight, $w$ \cite{aggarwal2011online,feldman2009online3,jaillet2014online}.
%The methodologies here provide key insight on our problem.

%% file: random_condensed.tex
In this section, we present the approximately greedy algorithm, Algorithm \ref{alg:smooth_greedy}, and analyze its competitive ratio in the random order model. Building on the work of Molinaro for online scheduling~\cite{Molinaro_2017}, we use a greedy algorithm for a smoothed version of our objective function and a restart procedure during the online allocation process to reduce the impact of correlations that arise in this input model.

A natural strategy for our problem is to allocate the arriving item to the least satisfied agent. However, one can show this strategy has too high an additive regret \cite{Gupta_2014, Molinaro_2017} as the change in our solution value can vary quickly from one iteration to the next. Instead, we use a \emph{smoothed} version of the greedy algorithm. The algorithm is designed as follows: we first define $\phi_\ve$ to be a re-scaled variant of the \textsc{LogSumExp} function that serves as a smoothed minimum. For the first half of the input stream, we select an allocation for each arriving item that maximizes the increase in our smoothed objective function. This stage can be thought of as approximately greedy with respect to the \emph{gradient} of $\phi_\ve$. After $\frac{m}{2}$ items have been allocated, we ``restart'' the allocation by maximizing the increase in our objective with respect to the $t > \frac{m}{2}$ allocations only. This restart procedure is essential for reducing the correlations that arise in sampling without replacement in the random order model since at each iteration of the allocation procedure, our decision depends on at most $\frac{m}{2}-1$ items. The pseudocode of this procedure is presented in Algorithm \ref{alg:smooth_greedy}.

\begin{algorithm}[h]
	\SetAlgoNoLine
	\KwIn{$0 < \ve < 1$, input stream $\mathcal{M}$ of $m$ items}
	% \KwOut{$(1-\ve)$-competitive \textsc{MaxMin} allocation}
	Define $\phi_\ve(u) = -\frac{1}{\ve} \ln \left(\sum_i e^{-\ve u_i}\right)$\;
	\For{$t=1$ to $m/2$}
	{
	    Select $\bx^t \in \Delta^n$ to maximize $\phi_\ve \left(\sum_{\tau = 1}^t \bV^{\tau} \right)$\;
	}
	\For{$t=m/2 + 1$ to $m$}
	{
	    Select $\bx^t \in \Delta^n$ to maximize $\phi_\ve \left(\sum_{\tau = \frac{m}{2} + 1}^t \bV^{\tau} \right)$\;
	}
	\caption{\textsc{Smooth Greedy With Restart}}
	\label{alg:smooth_greedy}
\end{algorithm}

\subsection{Algorithm Analysis}

In the analysis of the competitive ratio of Algorithm \ref{alg:smooth_greedy}, we will leverage several key facts about the smoothed minimum function $\phi_\ve$. This smoothness implies that the gradient nicely captures incremental increase in the objective function, thus allocating with respect to this produces an essentially greedy process, allowing us to follow the analysis of \cite{AgraDev_2014,devanur2009adwords,Molinaro_2017} to get the desired guarantees for the random-order model.%  \todo{make this more clear}. 
We now state our main result in Theorem \ref{thm:ROM}.

\begin{theorem} \label{thm:ROM}
    For any $\ve > 0$, Algorithm \ref{alg:smooth_greedy} guarantees in the random-order input model that the expected value of the allocation assigned to any agent is at least
    $$(1 - \ve)\cdot  \opt{} - O \left(\frac{\log n}{\ve}\right).$$
\end{theorem}
Note that this theorem immediately implies the following corollary since the additive regret term can be absorbed in the multiplicative error for sufficiently large $\opt$:
\begin{corollary}
    For any $\ve > 0$, Algorithm \ref{alg:smooth_greedy} has a competitive ratio of $(1-\e)$ for $\opt \ge \Omega(\log n/\ve^2)$.
\end{corollary}

In order to prove this theorem, we first show some properties of the smoothed minimum function $\phi_\ve$ utilized by Algorithm \ref{alg:smooth_greedy}. We will then prove some technical lemmas that will help establish the theorem.

Lemma \ref{lem:props} effectively defines the additive error with respect to our true objective, the agent-wise minimum, and stability of the smooth function following each allocation decision. As was shown in prior work, allocating with respect to the order statistics or even the $L_p^p$ norm produces either too high of a regret factor, or instability in the derivative value under small perturbations to the input value. As such, the smoothing and the following properties are critical to maintaining our regret and competitive ratio bounds.
\begin{lemma} \label{lem:props}
    For all $u \in \mathbb{R}^n$, $v \in [0,1]^n$, and $\ve > 0$, the function $\phi_\ve(x) = -\frac{1}{\ve}\ln \left( \sum_{i=1}^n e^{-\ve x_i} \right)$ satisfies the following:
    \begin{itemize}
        \item[(a)] $\min_i \{u_i\} - \frac{\ln n}{\ve} \leq \phi_\ve(u) < \min_i \{u_i\}$
        \item[(b)] $\nabla \phi_\ve (u + v) \in e^{\pm \ve} \cdot \nabla \phi_\ve (u)$
    \end{itemize}
    Furthermore, if $u_i \geq v_i$ for each $i \in [n]$, we have
    \begin{itemize}
        \item[(c)] $\phi_\ve(u-v) \leq \phi_\ve(u) - \phi_\ve(v)$ 
    \end{itemize}
\end{lemma}
The proof of these properties are deferred to Appendix \ref{sec:omitted} due to space constraints.

Now utilizing these two properties we can prove the following important bound on the inner product of the smoothed minimum's gradient. This will be used throughout our analysis to bound the incremental change in the objective \textsc{MaxMin} value after each allocation decision, and the summation of these changes can be seen as the accumulated ``reward'' at any given stage of the input stream. The proof is by direct integration of stability property $(b)$ and is thus deffered to Appendix \ref{sec:omitted}.

\begin{lemma} \label{lem:grad_bound}
    For $u \in \mathbb{R}^n$ and $v,v' \in [0,1]^n$, if $\phi_\ve(u + v) \geq \phi_\ve (u + v')$ then $\langle \nabla \phi_\ve(u), v \rangle \geq e^{-2\ve} \langle \nabla \phi_\ve(u), v' \rangle$.
\end{lemma}

We now follow in the intuition of Agrawal and Devanur \cite{AgraDev_2014} to prove Theorem \ref{thm:ROM} by bounding the incremental increase in our ``reward'' both before and after the restart at $t = \frac{m}{2}$. By implementing this restart in the allocation procedure, we segment the stream into two portions that have identical probabilistic guarantees and reduce the correlation between input elements to allow for an optimal competitive ratio and low additive regret.

\begin{proof}[Proof of Theorem \ref{thm:ROM}]
    We note again that in the approximately greedy procedure of Algorithm \ref{alg:smooth_greedy}, we are essentially seleting items greedily according to the \emph{gradient} of $\phi_\ve$ to maximize the incremental changes. Due to the restart at $m/2$, we define 
    \begin{gather*}
        \nabla^t = \nabla \phi_\ve\left(\sum_{\tau=1}^{t-1}\bV^{\tau}\right) \text{ for } t \leq \frac{m}{2} \text{ and } \nabla^t = \nabla \phi_\ve\left(\sum_{\tau=m/2 + 1}^{t-1}\bV^{\tau}\right) \text{ for } t > \frac{m}{2}
    \end{gather*}
    We now proceed by deriving a bound on $\min_i \{\sum_{\tau=1}^m\bV^{\tau}_{i}\}$ in terms of our smoothed-approximation function $\phi_\ve$, thus bounding the accumulated error by the algorithm when estimating our true objective. By the concavity of $\phi_\ve$ and Lemma \ref{lem:grad_bound}, we have that
    \begin{align*}
        \phi_\ve\left(\sum_{\tau=1}^t\bV^{\tau}\right) - \phi_\ve\left(\sum_{\tau=1}^{t-1}\bV^{\tau}\right) &\geq \left\langle \nabla \phi_\ve\left(\sum_{\tau=1}^t\bV^{\tau}\right), \bV^t \right\rangle \tag{Concavity} \\
        &\geq e^{-\ve} \left\langle \nabla \phi_\ve\left(\sum_{\tau=1}^{t-1}\bV^{\tau}\right), \bV^t \right\rangle \tag{Lemma \ref{lem:grad_bound}} = e^{-\ve} \langle \nabla^{t}, \bV^t \rangle.
    \end{align*}
    Without loss of generality we proceed by considering the first half of the input sequence ($t \leq \frac{m}{2}$). By summing this inequality over the input prior to the restart (from $t=1$ to $\frac{m}{2}$), we have
    \begin{align}
        \phi_\ve\left(\sum_{\tau=1}^{m/2}\bV^{\tau}\right) - \phi_\ve\left(0\right) = \phi_\ve\left(\sum_{\tau=1}^{m/2}\bV^{\tau}\right) + \frac{\ln n}{\ve} \geq e^{-\ve} \sum_{t=1}^{m/2} \langle \nabla^{t}, \bV^t \rangle \label{lem:olo_ineq}
    \end{align}
    Now, taking into account the allocation before and after the restart at $\frac{m}{2}$ and invoking Eq.~\ref{lem:olo_ineq} for each half of the input stream with the concavity of $\phi_\ve$, we obtain
    \begin{align*}
        \sum_{t=1}^m \langle \nabla^{t}, \bV^t \rangle &\leq e^\ve \left( \phi_\ve\left(\sum_{\tau=1}^{m/2}\bV^{\tau}\right) + \phi_\ve\left(\sum_{\tau=m/2 + 1}^{m}\bV^{\tau}\right) + \frac{2 \ln n}{\ve}\right) \tag{Ineq. \ref{lem:olo_ineq}}\\
        &\leq e^\ve \left( \phi_\ve\left(\sum_{\tau=1}^{m}\bV^{\tau}\right) + \frac{2 \ln n}{\ve}\right) \tag{Lemma \ref{lem:props}c} \\
        &\leq e^\ve \left( \min_i \{\sum_{\tau=1}^{m}\bV_i^{\tau}\} + \frac{2 \ln n}{\ve}\right) \tag{Lemma \ref{lem:props}a}
    \end{align*}
    % where the last inequality is a direct result of Lemma \ref{lem:props}.
    Lastly, by rearranging terms we can bound the element-wise minimum as
    \begin{align}
        \min_i \{\sum_{\tau=1}^{m}\bV_i^{\tau}\} \geq e^{-\ve} \left( \sum_{t=1}^m \langle \nabla^{t}, \bV^t\rangle \right) - \frac{2 \ln n}{\ve}. \label{ineq:min_bound}
    \end{align}
    Thus, the output of our algorithm will approximate the \emph{actual} objective function within a multiplicative factor of $e^{-\ve} \approx 1-\ve$ and an additive regret on the order of $\ln n / \ve$. 
    
    We now proceed to bound the gap between our algorithmic solution to that of the optimal \emph{offline} solution which, combined with the above error, will give the final result.
    As such, the final step of our analysis will be to take the expectation of this inequality to get the final bounds in Theorem \ref{thm:ROM}. Due to the restart and random order input model, the expected increase in maximal value over the first and second half is equivalent, $$\Ex\left[ \sum_{t=1}^{m/2} \langle \nabla^{t}, \bV^t\rangle\right] = \Ex\left[ \sum_{t=m/2+1}^{m} \langle \nabla^{t}, \bV^t\rangle\right]$$ so without loss of generality we need only bound the first half's value and apply this bound to both portions. %\todo{expand} 
    The benefit of this restart will yield a tolerable error as compared to the optimal offline solution in each half of the allocation procedure, rather than a continuously accumulating divergence between the two solutions: each arriving job's allocation is only dependent upon (at most) $\frac{n}{2} - 1$ other decisions. We leverage this randomness to obtain the final competitive guarantees \cite{molinaro2014geometry}. 
    
    Now, by the nature of greedy selection to maximize our objective function, we must have that in each iteration our algorithm's selection produces an increase in value that is at least as good as that of the optimal offline solution: %$\phi_\ve\left(\bV^t\right) \geq \phi_\ve(\bVo^t)$
    $\phi_\epsilon (\sum_{\tau=1}^t \bV^\tau)\geq \phi_\epsilon (\sum_{\tau=1}^{t-1} \bV^\tau+ \bVo^t)$
    for each $t$\footnote{The optimal offline algorithm may allocate in a manner that does not maximize the objective function's increase at iteration $t$ in anticipation of better allocation options later in the input sequence.}. Combining this with Lemma \ref{lem:grad_bound}, we must have 
    $$\langle \nabla^t, \bV^t \rangle \geq e^{-2\ve} \langle \nabla^t, \bVo^t \rangle$$
    Furthermore, since $\nabla \phi_\ve \in \ell_1^+$ \cite{gao2017properties} and $\nabla^t$ is independent of $\bV^t$, we prove the following purely probabilistic result that in Appendix \ref{sec:omitted} will be used to bound the above summations to derive the final result.

    \begin{lemma} \label{lem:hoeff}
    Consider a set of vectors $\{y^1, ..., y^m\} \in [0,1]^n$ and let $\{\bY^i\}_{i=1}^k$ be sampled without replacement. Let $\bZ \in \ell_1^+$ be a random vector that depends only on $\{\bY^i\}_{i=1}^{k-1}$. Then for all $\ve > 0$,
    $$\Ex\langle \bY^k, \bZ \rangle \geq e^{-\ve} \min_i \{\Ex\bY^k_i\} - \frac{\ln n}{\ve \left(m - k + 1 \right)}.$$
    \end{lemma}
    \noindent Using this lower bounding result, we can now bound the sum of rewards as
    \begin{align} \label{ineq.opt_lb}
        \Ex\left[ \langle \nabla^t, \bVo^t \rangle \right] \geq e^{-\ve} \min_i \{\Ex\left[ \bVo^t_i \right]\} - \frac{\ln n}{\ve(m - t + 1)}.
    \end{align}
    Adding this inequality over all $t \leq \frac{m}{2}$ in combination with $\min_i \{\Ex[ \bVo^t]_i\} = \frac{\opt{}}{m}$ we conclude that 
    \begin{align*}
        e^{2\ve} \cdot \Ex \sum_{t \leq \frac{m}{2}} \langle \nabla^t, \bV^t \rangle &\geq \Ex \sum_{t \leq \frac{m}{2}} \langle \nabla^t, \bVo^t \rangle \tag{Lemma \ref{lem:props}} \\
        &\geq e^{-\ve} \sum_{t \leq \frac{m}{2}} \min_i \{\Ex\left[ \bVo^t_i \right]\} - \sum_{t \leq \frac{m}{2}} \frac{\ln n}{\ve(m - t + 1)}  \tag{Ineq. \ref{ineq.opt_lb}}\\
        &= e^{-\ve} \left( \frac{\opt{}}{2}\right) - \frac{\ln n}{\ve}
    \end{align*}
    Due to the restart at $t = \frac{m}{2}$, we can extend the sum to $t = m$ by simply doubling the above RHS. Finally, invoking inequality (\ref{ineq:min_bound}), we see that
    $$\Ex\left[ \min_i \sum_{\tau=1}^m\bV^{\tau} \right] \geq e^{-4\ve} \opt{} - O \left(\frac{\log n}{\ve}\right) \geq \left(1-O(\ve)\right) \opt{} - O\left(\frac{\log n }{\ve}\right).$$
    by the Taylor approximation $e^{-x} \geq 1-x$. %and the assumption on the size of \opt{}. % \todo{randomized rounding of solution}
\end{proof}

%% file: rounding.tex
The previous algorithm produces an online fractional solution for the Santa Claus problem. We now show that using simple randomized rounding, we can convert this into an integer solution. %This rounding step requires a slightly stronger assumption of $\opt \ge \Omega\left(\frac{\log n}{\e^2}\right)$ compared to the assumption of  $\opt \ge \Omega\left(\frac{\log n}{\e}\right)$ for the fractional algorithm.

\begin{theorem}\label{thm:rounding}
   Fix any $\e > 0$. Given a $(1-\e)$competitive online fractional algorithm for the Santa Claus problem, there is an online (integral) algorithm whose competitive ratio is $(1-O(\e))$, provided $\opt\ge \Omega\left(\frac{\log n}{\e^2}\right)$.
\end{theorem}
\begin{proof}
    The algorithm is simply randomized rounding. If an item $j$ is allocated with fraction $x_{ij}$ to agent $i$ such that $\sum_{i=1}^n x_{ij} \le 1$ by the fractional solution, then we assign item $i$ to agent $j$ with probability $x_{ij}$. Note that since $v_{ij}\in [0, 1]$, the fractional value derived by an agent $i$ from an item $j$, given by $v_{ij} x_{ij}$ is also in $[0, 1]$. Thus, by Chernoff bounds, the probability that the total value of agent $i$ in the rounded assignment is less than $(1-\e)\sum_{j=1}^m v_{ij} x_{ij}$ is at most
    \[
        \exp\left(-\frac{\e^2}{3} \cdot\sum_{j=1}^m v_{ij} x_{ij}\right)
        \le \exp\left(-\frac{\e^2}{3} \cdot (1-\e) \cdot \opt\right)
        \le \frac{1}{n^2},
        \text{ for } \opt \ge \Omega\left(\frac{\log n}{\e^2}\right).
    \]
     Thus, with probability at least $1-\frac{1}{n^2}$, we have
    \[
        (1-\e)\sum_{j=1}^m v_{ij} x_{ij} \ge (1-\e)^2\cdot \opt > (1-2\e)\cdot \opt.
    \]
    It follows that the probability that the total value of agent $i$ in the rounded assignment is less than $(1-2\e)\cdot \opt$ is at most $\frac{1}{n^2}$. 
    Using the union bound over the $n$ agents, we get that the probability that the total value of {\em any} agent in the rounded assignment is less than $(1-2\e)\cdot \opt$ is at most $\frac1n$. Thus, the expected value of the objective is at least $(1-\frac1n)(1-2\e)\cdot \opt > (1-3\e)\cdot \opt$ for large enough $n$.
\end{proof}

%% file: lower.tex
We here present an impossibility result on the Santa Claus problem under the random order input models. 
\eat{
Following from Theorem \ref{thm:ad-upper}, it is clear that the online result for the Santa Claus problem is highly nontrivial and dependent upon the adversarial selections. For worst-case instances, we cannot hope to achieve a guarantee that is any greater than the most pessimistic possibility.

While in the random-order and \textsc{iid} input models the impact of such an adversary is diminished, we demonstrate that nearly optimal online results can only be achieved if the offline solution is large enough. This condition has been studied extensively for similar problems like bin-packing and covering which demonstrates a $\Omega (\log n / \ve^2)$ lower bound \cite{Agrawal_2014}, while other problems such as load balancing and makespan do not have a full consensus on the tightest bound. We here present the first result for the Santa Claus problem under the random-order assumption in Theorem \ref{thm:lower}. % We note that while other papers studying the \textsc{MaxMin} objective in the \textsc{iid} model assume a lower bound of $\Omega(\log n / \ve)$ or even the stronger $\Omega(\log n / \ve^2)$, none have proven the necessity and tightness of such a bound \cite{Gollapudi_2014,Kawase_2021}.
}

\begin{theorem}\label{thm:lower}
    For any $\gamma \in (0, 1)$, if a randomized algorithm \alg for the online Santa Claus problem with random order input satisfies 
    \[
        \alg \ge (1-\gamma) \cdot \opt,
    \]
    then $\opt \ge \frac{C \ln n}{\gamma}$ for some (absolute) constant $C > 0$.
\end{theorem}

We use the following construction from the proof of Theorem \ref{thm:ad-upper}. There are $n$ agents, of which $n-1$ are {\em private} agents and the remaining one is a {\em public} agent. Every private agent has $k$ distinct private items for which their valuation is $1$ each, and the valuation for every other agent $0$. In addition, there are $k$ public items, each of which has a valuation of $1$ for every agent. The optimal solution is to assign the private items to the corresponding private agents, and the public items to the public agent. Thus, $\opt = k$.

Now, when the items are presented in uniform random order, consider the first $\e$ fraction of presented items -- call this the $\e$-prefix. Our main claim is that with constant probability, the following statements both hold: 
\begin{enumerate}
    \item[(a)] there are about $\e$ fraction of public items in this $\e$-prefix
    \item[(b)] there is at least one type of private item that is missing from this $\e$-prefix.
\end{enumerate}
For property (a), we need the following concentration inequality from Devanur and Hayes~\cite{devanur2009adwords}:
\begin{lemma}[Lemma 3 in \cite{devanur2009adwords}]\label{lem:dh}
     Let $Y = (Y_1, \ldots, Y_m)$ be a vector of real numbers, and let $\e \in (0,1)$. Let $S$ be a random subset of $[m]$ of size $\e m$, and set $Y_S := \sum_{j\in S} Y_j$. Then, for every $\d \in (0,1)$,
     \[
        \Pr\left[|Y_S - \Ex[Y_S]| \ge \frac{2}{3} \|Y\|_\infty \ln \left(\frac{2}{\d}\right) + \|Y\|_2 \sqrt{2\e \ln \left(\frac{2}{\d}\right)}\right] 
        \le \d.
    \]
\end{lemma}
\noindent Property (a) concerning the fraction of public items in the $\e$-prefix follows almost immediately from the above lemma:
\begin{lemma}\label{lem:public}
     The probability that there are fewer than $\frac{5\e}{12}$ fraction of public items in the $\e$-prefix is at most $2 \exp{(-\e k / 8)}$.%$2 e^{-\frac{\e k}{8}}$.
\end{lemma}
\begin{proof}
    We invoke \ref{lem:dh} with the following setting of variables. Let $Y$ be a binary vector where $Y_i = 1$ for $i\in [k]$ and $Y_i = 0$ otherwise. ($Y_i$ is the indicator for whether an item is a public item.) $S$ represents the set of indices in the $\e$-prefix. Then, $Y_S$ counts the number of public items in the $\e$-prefix in a random ordering of the items. Clearly, $\Ex[Y_S] = \e k$, $\|Y\|_\infty = 1$, and $\|Y\|_2 = \sqrt{k}$. Finally, set $\d = 2 e^{-\frac{\e k}{8}}$, or more so $\ln \left(\frac{2}{\d}\right) = \frac{\e k}{8}$. Now, by \ref{lem:dh}, we have
    \[
        \Pr\left[|Y_S - \e k| \ge \frac{2}{3} \cdot \frac{\e k}{8} + \sqrt{k}\cdot \sqrt{2\e \cdot \frac{\e k}{8}}\right] 
        = \Pr\left[|Y_S - \e k| \ge \frac{7}{12} \cdot \e k\right] 
        \le 2 e^{-\frac{\e k}{8}}.
    \]
    and the lemma follows.
\end{proof}
\noindent Now, we show property (b) to demonstrate the non-existence of at least one type of private item in the $\e$-prefix.
\begin{lemma}\label{lem:private}
    The probability that all the $n-1$ types of private items appear in the $\e$-prefix is at most $\exp{\left(-(n-1) \cdot 4^{\frac{-\e k}{1-\e}}\right)}$.
\end{lemma}
\vspace{-4mm}
\begin{proof}
    Fix a type of private item, say those of type-$j$, i.e. only agent $j$ (for some $j\in [n-1]$) has unit value for this type of item while all other agents have $0$ value. First, we bound the probability that no item of type-$j$ appears in the $\e$-prefix. To do this, note that this probability can be written, using the chain rule for conditional probabilities, as the product (over $i$ from $1$ to $\e nk$) of the probabilities of the $i$th item in the $\e$-prefix not being a type-$j$ item under the condition that the first $i-1$ items were not type-$j$ items either. Clearly, this probability, for any $i\le \e nk$, is at least $1-\frac{1}{(1-\e)n}$ since there are $k$ items of type-$j$ among at most $(1-\e)nk$ items overall after the conditioning. Thus, the probability that no item of type-$j$ appears in the $\e$-prefix is at least
    \[
        \left(1-\frac{1}{(1-\e)n}\right)^{\e nk} \ge 4^{\frac{-\e}{1-\e}\cdot k} \text{ by choosing } n \ge \frac{2}{1-\e}.
    \]
    Denote $p = 4^{\frac{-\e}{1-\e}\cdot k}$. Consider the events that at least one item of type-$j$ appears in the $\e$-prefix. These events are negatively correlated and therefore, the probability that at least one item of type-$j$ appears in the $\e$-prefix for every $j\in [n-1]$ is at most
    \[
        (1-p)^{n-1} 
        \le \left(1-4^{\frac{-\e}{1-\e}\cdot k}\right)^{n-1}
        \le e^{-4^{\frac{-\e}{1-\e} k}(n-1)}.
    \]
\end{proof}
\vspace{-2mm}
Now, by setting $k = \frac{1-\e}{2\e} \cdot \lg (n-1)$ we obtain
\eat{
\begin{equation}\label{eq:k}
    k = \frac{1-\e}{2\e} \cdot \lg (n-1).
\end{equation}
}
% Then, 
\[
    (n-1) \cdot 4^{\frac{-\e}{1-\e} k}
    = (n-1) \cdot 2^{-\frac{2\e}{1-\e} \cdot \frac{1-\e}{2\e}\cdot \lg (n-1)}
    = (n-1) \cdot 2^{-\lg (n-1)}
    = 1,
\]
and furthermore,
\[
    \frac{\e k}{8} = \frac{1-\e}{16}\cdot \lg(n-1).
\]
Plugging these expressions into \ref{lem:public} and \ref{lem:private}, we get that the probability of every private item type appearing in the $\e$-prefix is at most $\frac1e$ and the probability of fewer than $\frac{5\e}{12}$ fraction of public items appearing in the $\e$-prefix is at most $2 e^{-\frac{1-\e}{16}\cdot \lg(n-1)}$. The latter probability can be further simplified to
\[
    2 e^{-\frac{1-\e}{16}\cdot \lg(n-1)}
    = 2 e^{-\frac{1-\e}{16}\cdot \frac{\ln(n-1)}{\ln 2}}
    = \left(\frac{2}{n-1}\right)^{\frac{1-\e}{16\ln 2}}
    < \frac1e \text{ for large enough } n.
\]
Using the union bound over these two events, we conclude that with probability at least $1-\frac2e$, there is at least one missing private item type in the $\e$-prefix, and also at least $\frac{5\e}{12}$ fraction of private items appear in the $\e$-prefix. In this case, in the $\e$-prefix, the algorithm cannot distinguish between the private agent whose item type is missing and the public agent. Thus, in expectation (over the randomness of the algorithm), at least half the public items in the $\e$-prefix are assigned to the private agent whose item type is missing. As a consequence, at least $\frac{5\e}{24}$ fraction of the public items are not allocated to the public agent in the entire algorithm, which means that the total valuation of the public agent is at most $\left(1 - \frac{5\e}{24}\right) k$. Thus, in order to guarantee 
\[
    \alg \ge (1-\gamma)\cdot \opt,
\]
we need $\gamma \ge \frac{5\e}{24}$, i.e., $\frac1\e \ge \frac{5}{24\gamma}$. This implies by the set value of $k$ that %\Cref{eq:k} that
\[
    \opt
    = k 
    = \frac{1-\e}{2\e}\cdot \lg (n-1)
    = \frac{\frac1\e-1}{2}\cdot \lg (n-1)
    \ge \frac{\frac{5}{24\gamma}-1}{2}\cdot \lg (n-1)
    \ge C\cdot \frac{\ln n}{\gamma} \text{ for some constant } C.
\]
This completes the proof of \ref{thm:lower}.\qed

%% file: conclusion.tex
In the current work, we have presented impossibility results for the online Santa Claus problem in adversarial and random order input models, as well as a near optimal algorithm for the random order setting. These results effectively address the necessary assumptions on the problem to obtain optimal online solutions, and furthermore obtain these results via simplistic algorithms. % Though the time complexity of these algorithms is not presented here, both are fast and do not rely on solving an LP at any stage as is the case for most results on the more complex bin-packing or AdWords problems \cite{Devanur_2011,devanur2009adwords,kesselheim2014primal}
Since the \textsc{MaxMin} objective function is the dual of \textsc{MinMax} used in the makespan and load balancing literature, we hope that the impossibility results presented here can be carried over to address the remaining open questions regarding their tightest additive terms in competitive ratio analysis.

% A necessary future direction is to bridge the gap between our upper and lower bounding results, which we demonstrate to differ by a factor of $\ve$. 
% Future work should extend the results to the case of \emph{weighted} or asymmetric agents. This setting is important in the online advertising space as normally competing agents have some effective budgets that grant them priority in the allocation process. Furthermore, it would be of great interest to examine the setting in which arriving item values are only revealed after the allocation decision is made -- an experts problem variant on the presented work. Though this extension may introduce too much complexity to handle with an optimal competitive ratio, it is nonetheless a logical progression. 

%% file: omitted_proofs.tex
\subsection{Online Santa Claus with Adversarial Arrival Order}\label{sec:adversarial}\input{adversarial}

\subsection{Online Algorithm Lemmas}
\begin{proof}[Proof of Lemma \ref{lem:props}]
Property $(a)$ is an established property of the \textsc{LogSumExp} function \cite{blanchard2021accurately,nielsen2016guaranteed} but we present the proof here for completeness. We start by exponentiating the input, summing over all elements and applying the logarithm to the resultant bounds.
\begin{align*}
    &\exp{ \max_i \{-\ve u_i \}} \leq \sum_{i=1}^n \exp{-\ve u_i} < n \exp{\max_i \{-\ve u_i \}} \\
    &\overset{(I)}{\iff} \max_i \{- \ve u_i \} \leq \ln \sum_{i=1}^n \exp{-\ve u_i} < \max_i \{- \ve u_i \} + \ln n \\
    &\overset{(II)}{\iff} -\max_i \{- \ve u_i \} > -\ln \sum_{i=1}^n \exp{-\ve u_i} \geq -\max_i \{- \ve u_i \} - \ln n \\
    &\overset{(III)}{\iff} \min_i \{\ve u_i \} > -\ln \sum_{i=1}^n \exp{-\ve u_i} \geq \min_i \{\ve u_i \} - \ln n
\end{align*}
Where $(\textsc{i})$ is the result of taking the logarithm, $(\textsc{ii})$ is a negation on the inequalities and $(\textsc{iii})$ is by property of the maximum. Now, since $\ve > 0$, the result follows from simple algebraic manipulation.
\begin{align*}
    &\min_i \{\ve u_i \} > -\ln \sum_{i=1}^n \exp{-\ve u_i} \geq \min_i \{\ve u_i \} - \ln n \\
    &\overset{(IV)}{\iff} \ve \min_i \{u_i \} > -\ln \sum_{i=1}^n \exp{-\ve u_i} \geq \ve \min_i \{u_i \} - \ln n
\end{align*}
\begin{align*}
    &\overset{(V)}{\iff} \min_i \{u_i \} > \frac{-1}{\ve} \ln \sum_{i=1}^n \exp{-\ve u_i} \geq \min_i \{u_i \} - \frac{\ln n}{\ve} \\
    &\overset{(VI)}{\iff}\min_i \{u_i \} > \phi_\ve(u) \geq \min_i \{u_i \} - \frac{\ln n}{\ve}
\end{align*}
where $(\textsc{iv})$ follows from positive scalar multiplication within a minimum, $(\textsc{v})$ by dividing through by $\ve$ and $(\textsc{vi})$ is merely the definition of our smoothing function $\phi_\ve$. This verifies the desired property. 

To prove $(b)$, we first calculate the partial derivative of the smoothed minimum function 
% \debmalya{$\frac{\partial}{\partial x_i}$ or $\frac{\partial}{\partial u_i}$?}
$$\frac{\partial}{\partial x_i} \phi_\ve (u) = \frac{e^{-\ve u_i}}{\sum_{j=1}^n e^{-\ve u_j}}$$
and now, using $u_i \geq 0$ and $v_i \in [0,1]$ we derive
% \debmalya{Again, $\frac{\partial}{\partial x_i}$ or $\frac{\partial}{\partial u_i}$?}
\begin{gather*}
    \frac{e^{-\ve}e^{-\ve u_i}}{\sum_{j=1}^n e^{-\ve u_j}} < \frac{e^{-\ve (u_i + v_i)}}{\sum_{j=1}^n e^{-\ve (u_j + v_j)}} <  \frac{e^{-\ve u_i}}{e^{-\ve} \sum_{j=1}^n e^{-\ve u_j}}\\
    e^{-\ve}\frac{\partial}{\partial x_i} \phi_\ve(u) < \frac{\partial}{\partial x_i} \phi_\ve(u+v)  < e^{\ve} \frac{\partial}{\partial x_i} \phi_\ve(u) 
\end{gather*}
Therefore, we have property $(b)$.

Lastly, to prove $(c)$ we first invoke the definition $\phi_\ve$:
$$\frac{-1}{\varepsilon} \ln \left( \sum_{i=1}^n e^{-\varepsilon(x_i - y_i)} \right) \leq \frac{-1}{\varepsilon} \left( \ln \left( \sum_{i=1}^n e^{-\varepsilon x_i} \right) - \ln \left( \sum_{i=1}^n e^{-\varepsilon y_i} \right) \right).$$
This statement is equivalent to
$$\ln \left( \sum_{i=1}^n e^{-\varepsilon(x_i - y_i)} \right) \geq \ln \left( \sum_{i=1}^n e^{-\varepsilon x_i} \right) - \ln \left( \sum_{i=1}^n e^{-\varepsilon y_i} \right)$$
which, by exponentiating both sides, yields
$$\sum_{i=1}^n e^{-\varepsilon(x_i - y_i)} \geq \left(\sum_{i=1}^n e^{-\varepsilon x_i} \right) \cdot \left( \sum_{i=1}^n e^{-\varepsilon y_i} \right)^{-1} \Rightarrow \sum_{i=1}^n e^{-\varepsilon(x_i - y_i)} \cdot \left( \sum_{i=1}^n e^{-\varepsilon y_i} \right) \geq \left(\sum_{i=1}^n e^{-\varepsilon x_i} \right)$$
and expansion of the left-hand side verifies the claim.
\end{proof}

\begin{proof}[Proof of Lemma \ref{lem:grad_bound}] By direct integration and the stability property $(b)$, we see
\begin{align*}
    \phi_\ve (u+v) &= \phi_\ve (u) + \int_0^1 \langle \nabla \phi_\ve (u + \alpha v), v \rangle d\alpha \\
    &\in \phi_\ve(u) + e^{\pm \ve} \langle \nabla \phi_\ve (u), v\rangle
\end{align*}
Therefore, we can further show
\begin{align*}
    \langle \nabla \phi_\ve (u), v \rangle &\geq e^{-\ve} [\phi_\ve (u+v) - \phi(u)] \\
    &\geq e^{-\ve} [\phi_\ve (u + v') - \phi (u)] \\
    &\geq e^{-2\ve} \langle 
    \nabla \phi_\ve (u), v' \rangle
\end{align*}
where the first and last inequality is a direct result of Lemma \ref{lem:props}, and the second is from assumption on the inputs.
\end{proof}

    \begin{proof}[Proof of Lemma \ref{lem:hoeff}]
    Let $\mu = \frac{1}{m} \sum_t y^t$ be the average of the set of vectors and to simplify notation we additionally break apart the expectation and let %$\Ex\langle \bY^k, \bZ \rangle = \Ex \left[ \Ex \left[ \langle \bY^k, \bZ \rangle | \bY^1, ..., \bY^{k-1} \right] \right]$.
    $\Ex\langle \bY^k, \bZ \rangle = \Ex \Ex_{k-1}\langle \bY^k, \bZ \rangle$ where $\Ex_{k-1}$ denote the expectation conditioned on $\bY^1, ..., \bY^{k-1}$.
    Note that since $\bZ$ is a unit-vector in $\ell_1^+$, an innerproduct of this vector with $\bY^k$ is simply a weighted sum of the latter's elements. Thus, we have
    \begin{align} \label{ineq.inner_prod}
        \Ex_{k-1}\langle \bY^k, Z \rangle%&=  \langle \Ex\left[\bY^k | \bY^1, ..., \bY^{k-1} \right], Z \rangle \\
        \geq \min_i \{ \Ex \left[ \bY^k | \bY^1, ..., \bY^{k-1}\right] \}.
    \end{align}
    Additionally, by nature of the sampling set and the procedure of sampling without replacement, we have the conditional expectation
    $$\Ex\left[ \bY^k | \bY^1, ..., \bY^{k-1} \right] = \frac{m\mu - (\bY^1 + ... + \bY^{k-1})}{m - (k - 1)}$$ and further note that $m\mu - (\bY^1 + ... + \bY^{k-1})$ has the \emph{same} distribution as $\bY^1 + ... + \bY^{m-(k-1)}$. This concretely gives us the simplifying equivalences
    $$\Ex\left[ \bY^k | \bY^1, ..., \bY^{k-1} \right] = \frac{m\mu - (\bY^1 + ... + \bY^{k-1})}{m-(k-1)} = \frac{\sum_{t=1}^{m-(k-1)} \bY^t}{m-(k-1)}.$$
    We now return to the inequality bound of (\ref{ineq.inner_prod}) and, using the above equivalences, obtain 
    \begin{align*}
        \Ex \left[ \min_i \left\{ \sum_{t=1}^{m-(k-1)} \bY_i^t \right\} \right] &\geq \Ex \left[ \phi_\ve \left(\Sigma_{t=1}^{m-(k-1)} \bY^t \right) \right] \tag{Lemma \ref{lem:props}}\\
        &= \Ex \left[\frac{-1}{\ve} \ln \left( \sum_i \exp{\left(-\ve \Sigma_{t=1}^{m-(k-1)} \bY_i^t\right)} \right) \right]\\
        &\geq \frac{-1}{\ve} \ln \left( \sum_i \exp{\left(-\ve \Ex \left[ \Sigma_{t=1}^{m-(k-1)} \bY_i^t \right]\right)} \right) \tag{Jensen's Ineq.} \\
        &\geq e^{-\ve} \min_i \left\{ \Ex [\Sigma_{t=1}^{m-(k-1)} \bY_i^t] \right\} - \frac{\ln n}{\ve} \tag{Lemma \ref{lem:props}} \\
        &\geq e^{-\ve} \min_i \left\{ (m-(k-1)) \Ex \bY_i^t \right\} - \frac{\ln n}{\ve}
    \end{align*}

    Finally combining the above results, we obtain
    \begin{align*}
        \frac{1}{m-(k-1)} \left(e^{-\ve} \min_i \{ (m-(k-1)) \Ex \bY^t \} - \frac{\ln n}{\ve} \right) &\leq \Ex \left[ \min_i \{\Ex \left[\bY^k_i | \bY^1, ..., \bY^{k-1} \right] \} \right] \\
        &\leq \Ex \langle \bY^t, \bZ \rangle
    \end{align*}
    Thus, after rearranging terms, this completes the lemma.
    \end{proof}

%% file: adversarial.tex
% First, we prove \Cref{thm:ad-upper}.
\begin{proof}[Proof of \Cref{thm:ad-upper}]
    We consider the following instance: of the $n$ agents, $n-1$ are \emph{private} agents each having $k$ private items that have valuation $1$ for the corresponding private agent and $0$ for every other agent. The $n$-th agent is a \emph{public} agent and there are $k$ public items that have a valuation of $1$ for all the $n$ agents. % A graphical depiction of these two item types is presented in Figure \ref{fig:item_types}.
    
    The optimal solution assigns the private items to the corresponding private agents and the public items to the public agent. Thus, $\opt = k$. In the online instance, the adversary chooses to present all the public items before the private items. Since all the agents look identical before the arrival of the first private item, the public agent gets no more than $k/n$ items in expectation for any algorithm (this adversarial input is comparable to our toy example discussed in Figure \ref{fig:hard_ex}). Since none of the remaining private items can be allocated to the one public agent, their bundle value can no longer increase beyond $k/n$. The theorem follows.
\end{proof}

\begin{theorem} \label{thm:ad-algo}
    In the adversarial setting, for any $\e \in (0, 1)$, there is an algorithm for the online Santa Claus problem that has a competitive ratio of $(1-\e)\cdot\frac1n$ for $\opt > C\cdot \frac{n\ln n}{\e^2}$ for a large enough constant $C$.
\end{theorem}
\begin{proof}
    The algorithm is to simply assign every item uniformly at random among all the $n$ agents. Note that for any fixed agent, its expected value is at least $\frac1n \cdot \opt$. By Chernoff bounds, the probability that its total value is less than $(1-\frac\e2)\frac1n \cdot \opt$ is given by 
    $\exp(-\frac12\cdot \frac{\e^2}{4}\cdot \frac1n \cdot \opt) < \frac{\e}{2n}$ for $\opt > C\cdot \frac{n \ln n}{\e^2}$ for a sufficiently large constant $C$. By union bound over all the $n$ agents, the probability that {\em any} agent's overall value is less than $(1-\frac\e2)\frac1n\cdot\opt$ is at most $\frac\e2$. Thus, the expected competitive ratio is at least $(1-\frac\e2)(1-\frac\e2)\cdot \frac1n > (1-\e)\cdot \frac1n$.
\end{proof}